\documentclass[english,review,12pt,a4paper]{article}

\usepackage[utf8]{inputenc}
\usepackage[T1]{fontenc}
\usepackage{babel}

\usepackage{amsthm}
\newtheorem{theorem}{Theorem}
\newtheorem{lemma}{Lemma}
\newtheorem{corollary}{Corollary}

\newtheorem{proposition}{Proposition}
\newtheorem{ex}{Example}

\newenvironment{example-cont}[1]{\bigskip\noindent\textbf{Example~\ref{#1}.~(cont.)\hspace{\labelsep}}}{\bigskip\noindent}

\usepackage{macros}
\usepackage[noblocks]{authblk}

\usepackage[stable,multiple]{footmisc}
\usepackage{manyfoot}
\DeclareNewFootnote{email}[alph]
\DeclareNewFootnote{support}[fnsymbol]
\stepcounter{footnotesupport}

\title{%
	Weight-Reducing Turing Machines%
	\footnote{%
		This work contains, in an extended form, some material and results
		which were previously presented in a preliminary form
		in conference papers~\cite{Pru14} and~\cite{GPPP18}.%
	}%
}

\author[,1]{Bruno Guillon%
	\protect\footnoteemail{bruno.guillon@uca.fr}%
}

\author[,2]{Giovanni Pighizzini%
	\protect\footnoteemail{pighizzini@di.unimi.it}%
	\protect\footnotesupport{Partially supported by Gruppo Nazionale per il Calcolo Scientifico (GNCS-INdAM).}%
}
\author[,2]{Luca Prigioniero%
	\protect\footnoteemail{prigioniero@di.unimi.it}%
}

\author[,3]{Daniel Pr\r u\v sa%
	\protect\footnoteemail{prusapa1@cmp.felk.cvut.cz}%
	\protect\footnotesupport{Supported by the Czech Science Foundation, grant 19-21198S.}%
}

\affil[1]{LIMOS, Université Clermont-Auvergne, France}
\affil[2]{Dipartimento di Informatica, Universit\`a degli Studi di Milano, Italy}
\affil[3]{Faculty of Electrical Engineering, Czech Technical University, Prague}

\date{}

\begin{document}
\maketitle
\begin{abstract}
	\noindent
	It is well-known that one-tape Turing machines working in linear time are no more powerful than finite automata, 
	namely they recognize exactly the class of regular languages.
	We prove that it is not decidable if a one-tape machine works in linear time,
	even if it is deterministic and restricted to use only the portion of the tape which initially contains the input.
	This motivates the introduction of a constructive variant of one-tape machines,
	called weight-reducing machine, and the investigation of its properties.
	We focus on the deterministic case.
	In particular, we show that,
	paying a polynomial size increase only,
	each weight-reducing machine
	can be turned into a halting one that works in linear time.
	Furthermore each weight-reducing machine
	can be converted into equivalent nondeterministic and deterministic finite automata
	by paying exponential and doubly-exponential increase in size, respectively.
	These costs cannot be reduced in general.
\end{abstract}
\vbox{}



\section{Introduction}
\label{sec:intro}

The characterization of classes of languages
in terms of recognizing devices
is a classical topic in formal languages and automata theory.
The bottom level of the Chomsky hierarchy,
\ie, the class of type~3 or regular languages,
is characterized in terms of deterministic and nondeterministic finite automata (\dfa\s and \nfa\s, respectively).
The top level of the hierarchy,
\ie, type~0 languages,
can be characterized by Turing machines
(in both deterministic and nondeterministic versions),
even in the one-tape restriction,
namely with a unique infinite or semi-infinite tape containing,
at the beginning of the computation,
only the input,
and whose contents can be rewritten to store information.

Considering machines that make a restricted use of space or time,
it is possible to characterize other classes of the hierarchy.
On the one hand,
if the available space is restricted only to the portion of the tape
which initially contains the input
and nondeterministic transitions are allowed,
the resulting model,
known as \emph{linear-bounded automaton},
characterizes type~$1$ or context-sensitive languages~\cite{Kur64}.
(The power does not increase when the space is \emph{linear} in the input length.)
On the other hand,
when the length of the computations,
\ie, the time, is linear in the input length,
one-tape Turing machines are no more powerful than finite automata,
namely they recognize only regular languages,
as proved by Hennie in 1965~\cite{Hen65}.%
\footnote{%
	Actually, the model considered by Hennie was deterministic.
	Several extensions of this result,
	including that to the nondeterministic case
	and greater time lower bounds for nonregular language recognition,
	have been stated in the literature~\cite{Tr64,Har68,Mic91,Pig09,TYL10}.%
}

The main purpose of this paper is the investigation of some fundamental properties of several variants of one-tape Turing machines working in linear time.
We now give an outline of the motivation for this investigation and of the results we present.
(\Cref{fig:machines} summarizes the models we are going to discuss and their relationships.)

\begin{figure}[tb]
	\centering
	\includegraphics[width=\textwidth]{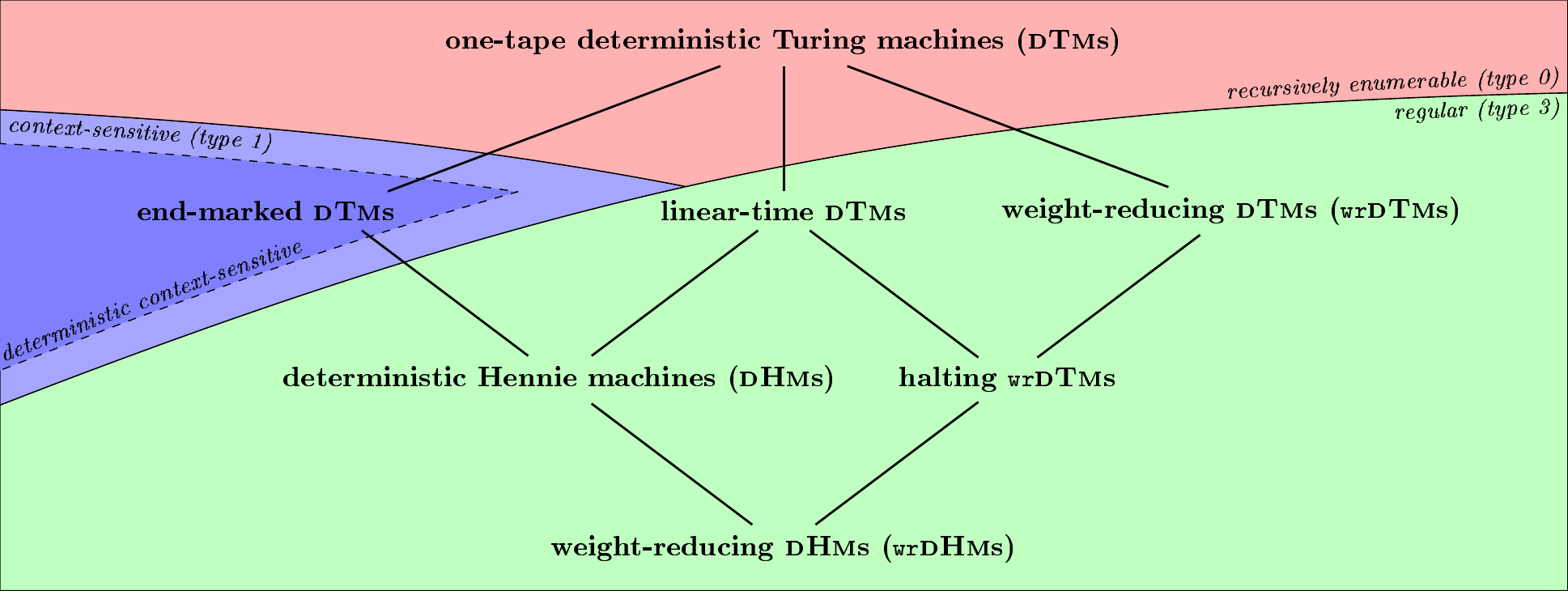}
	\caption{%
		Variants of one-tape deterministic Turing machines
		and their expressive power confronted with the Chomsky hierarchy.
		In particular,
		end-marked \dtm\s are known as \emph{deterministic linear bounded automata} in the literature,
		and recognize the so-called \emph{deterministic context-sensitive languages},
		a subclass of context-sensitive languages,
		see~\cite{Wal70}.
		It is still unknown if such inclusion is strict.
	}
	\label{fig:machines}
\end{figure}

A natural question concerning models
that share the same computational power
is the comparison of the sizes of their descriptions.
In this respect,
one could be interested in comparing one-tape Turing machines working in linear time
with equivalent finite automata.

Here, we prove that there exists no recursive function
bounding the size blowup resulting from the conversion
from one-tape Turing machines working in linear time
into equivalent finite automata.
Hence, one-tape linear-time Turing machines
can be arbitrarily more succinct than equivalent finite automata.
Furthermore, it cannot be decided whether or not a one-tape Turing machine works in linear time.%
\footnote{%
	For the sake of completeness,
	we mention that it is decidable whether or not
	a machine makes at most~$cm+d$ steps on input of length~$m$,
	for any \emph{fixed}~$c,d>0$~\cite{Gaj15}.%
}
These results remain true in the restricted case of \emph{end-marked machines}, 
namely one-tape deterministic Turing machines that do not have any extra space, besides the tape
portion which initially contains the input.
Deterministic end-marked machines working in linear time will be called \emph{Hennie machines}.

To overcome the above-mentioned ``negative'' results, 
we consider a syntactical restriction on one-tape deterministic Turing machines,
thus introducing \emph{weight-reducing Turing machines}. 
This restriction aims to enforce the machine to work in linear time,
by making any tape cell rewriting decreasing
according to some fixed partial order on the working alphabet.
However, due to the unrestricted amount of available tape space,
these devices can have non-halting computations.
Nevertheless, they work in linear time as soon as they are halting.
Indeed, we prove that each computation either halts within a time which is linear in the input length, or is infinite.
In the paper we show that it is possible to decide whether a weight-reducing Turing machine is halting.
As a consequence, it is also possible to decide whether it works in linear time.
Furthermore, with a polynomial size increase,
any such machine can be made halting whence working in linear time.
Our main result is that the tight size cost of converting
a weight-reducing Turing machine into a \dfa
is double exponential.
This cost reduces to a simple exponential
when the target device is an \nfa.

Considering end-marked Turing machines satisfying the weight-reducing syntactical restriction,
we obtain \emph{weight-reducing Hennie machines}. 
These devices do not allow infinite computations whence always work in linear time.
The above-stated double exponential blowup is easily extended to them.

The paper is organized as follows.
\Cref{sec:prel} presents the fundamental notions and definitions,
included those related to the computational models we are interested in.
\Cref{sec:undec} is devoted to prove the above-mentioned undecidability and non-recursive trade-off results
concerning Hennie machines.
In \cref{sec:wrtm-versus-wrhm},
after proving that it can be decided if a deterministic Turing machine is weight-reducing,
we describe a procedure that,
given a linear-time machine
together with the coefficient of its linear bound on time,
turns it into an equivalent weight-reducing machine.
Furthermore, we present a simulation of weight-reducing machines by finite automata, studying its size cost.
In \cref{sec:spaceTimeHalt}
we show how to decide if a weight-reducing machine halts on any input and if it works in linear time.
We also prove that by a polynomial increase in size,
each weight-reducing machine can be transformed into an equivalent one
which always halts and which works in linear time.

\section{Preliminaries}%
\label{sec:prel}

In this section we recall some basic definitions
and notations.
We also describe the main computational models considered in the paper.

We assume the reader familiar
with notions from formal languages
and automata theory
(see, \eg,~\cite{HU79}).
Given a set~$S$,
\defd{$\card S$} denotes its cardinality
and~\defd{$\powerset S$} the family of all its subsets.
Given an alphabet~\alphab,
\defd{$\length w$} denotes the length of a string~$w\in\alphab^*$,
\defd{$\ith w$} the~$i$\=/th symbol of~$w$, $i=1,\ldots,\length w$,
and~\defd{\emptyword} denotes the empty string.
\smallbreak

The main computational model we consider is the \defd{deterministic one-tape Turing machine} (\dtm).
Such a machine
is a tuple~$\struct{Q,\ialphab,\walphab,\delta,q_0,F}$
where~$Q$ is the \defd{set of states},
\ialphab is the \defd{input alphabet},
\walphab is the \defd{working alphabet}
including symbols of \ialphab
and the special \defd{blank symbol}, denoted by~$\tblank$,
that cannot be written by the machine,
$q_0\in Q$ is the \defd{initial state},
$F\subseteq Q$ is the \defd{set of final states},
and~%
$\delta: Q\times\walphab\to Q\times(\walphab\setminus\set\tblank)\times\moveset$
is the partial \defd{deterministic transition function}.
In one step,
depending on its current state~$p$
and on the symbol~$\sigma$ read by the head,
a \dtm changes its state to~$q$,
overwrites the corresponding tape cell with~$\tau$
and moves the head one cell to the left or to the right
according to~$\xmove = \lmove$ or~$\xmove = \rmove$, respectively,
if~$\delta(p,\sigma)=(q,\tau,\xmove)$.
Since~$\delta$ is partial,
it may happen that no transition can be applied.
In this case, we say that the machine \defd{halts}.
At the beginning of computation
the input string~$w$ resides on a segment of a bi-infinite tape,
called \defd{initial segment},
and the remaining infinitely many cells contain the blank symbol.
The computation over~$w$
starts in the initial state
with the head scanning the leftmost symbol of~$w$ if~$w\neq\emptyword$
or a blank tape cell otherwise.
The input is \defd{accepted}
if the machine eventually halts in a final state.
The language accepted by a \dtm~\mA is denoted by~\langof\mA.

Let~$\mA=\struct{Q,\ialphab,\walphab,\delta,q_0,F}$ be a \dtm.
A \defd{configuration} of~\mA
is given by the current state~$q$,
the tape contents,
and the position of the head.
If the head is scanning a non-blank symbol,
we describe it by~\defd{\config zqu}
where~$zu\in\walphab^*$ is the finite non-blank contents of the tape,
$u\neq\emptyword$,
and the head is scanning the first symbol of~$u$.
Otherwise, we describe it by~\defd{\config{}{q}{\tblank z}} or~\defd{\config zq{}}
according to whether the head is scanning
the first blank symbol to the left or to the right of the non-blank tape contents~$z$, respectively.
If the device may enter a configuration~$\config{z'}{q'}{u'}$
from a configuration~$\config zqu$ in one step,
we say that~$\config{z'}{q'}{u'}$ is \defd{a successor} of~$\config zqu$,
denoted~$\config zqu\yields\config{z'}{q'}{u'}$.
A \defd{halting configuration} is a configuration that has no successor.
The reflexive and transitive closure of~$\yields$
is denoted by~\defd{\yields*}.
On an input string~$w\in\ialphab^*$,
the \defd{initial configuration} is~$\config{}{q_0}{w}$.
An \defd{accepting configuration} is a halting configuration~$\config{z}{q_f}{u}$
such that~$q_f$ is a final state of the machine.
A \defd{computation} is a (possibly infinite) sequence of successive configurations.
It is \defd{accepting} if it is finite,
its first configuration is initial,
and its last configuration is accepting.
Therefore,~%
\[
	\langof\mA=\set{
		w\in\ialphab^*
		\mid
			\config{q_0}{w}{}
			\yields*
			\config z{q_f}u,\,
			\text{where $q_f\in F$ and $\config z{q_f}u$ is halting}
		}%
	\text.
\]
\bigbreak

In the paper we consider the following restrictions of \dtm\s (see \cref{fig:machines}).
\begin{description}
	\item[End-marked machines.]
		We say that a \dtm is \defd{end-marked},
		if at the beginning of the computation the input string is surrounded
		by two special symbols belonging to~$\walphab$,
		\defd{\lend} and~\defd{\rend} respectively,
		called the \emph{left} and the \emph{right endmarkers},
		which can never be overwritten,
		and that prevent the head to fall out the tape portion
		that initially contains the input.
		Formally, for each transition~$\delta(p,\sigma)=(q,\tau,\xmove)$,
		$\sigma=\lend$ (\resp,~$\sigma=\rend$) implies%
		~$\tau=\sigma$ and~$\xmove=\rmove$ (\resp,~$\xmove=\lmove$).
		This is the deterministic restriction
		of the well-known \emph{linear-bounded automata}%
		~\cite{Kur64}.
		For end-marked machines,
		the initial configuration on input~$w$
		is~$\config{}{q_0}{\lend w\rend}$.
	\item[Weight-reducing Turing machines.]
		A \dtm is \defd{weight-reducing} (\defd\wrdtm),
		if there exists a partial order~$<$ on~\walphab
		such that each rewriting is decreasing,
		\ie, $\delta(p,\sigma)=(q,\tau,\xmove)$ implies~$\tau<\sigma$.
		By this condition, in a \wrdtm the number of visits to each tape cell
		is bounded by a constant. However, one \wrdtm could have non-halting
		computations which, hence, necessarily visit infinitely many tape cells.
	\item[Linear-time Turing machines.]
		A \dtm is said to be \defd{linear-time} if
		over each input~$w$,
		its computation halts within~$\bigoof{\length w}$ steps.
	\item[Hennie machines.]
		A \defd{Hennie machine} (\defd\dhm)  
		is a linear-time \dtm which is,
		furthermore, end-marked.
	\item[Weight-Reducing Hennie machines.]
		By combining previous conditions,
		\defd{weight-reducing Hennie machines} (\defd\wrdhm) are defined as particular \dhm,
		for which there exists an order~$<$
		over~$\walphab\setminus\set{\lend,\rend}$
		such that $\delta(p,\sigma)=(q,\tau,\xmove)$ implies~$\tau<\sigma$
		unless $\sigma\in\set{\lend,\rend}$.
		Observe that each end-marked \wrdtm can execute 
		a number of steps which is at most linear in the length of the input.
		Hence,
		end-marked \wrdtm
		are necessarily weight-reducing Hennie machines.
\end{description}
We also consider finite automata.
We briefly recall their definition.
A \defd{nondeterministic finite automaton}
(\defd\nfa)
is a computational device
equipped with a finite control and a finite read-only tape
which is scanned by an input head in a one-way fashion.
Formally,
it is defined as a quintuple~$\mA=\struct{Q,\alphab,\delta,q_0,F}$,
where~$Q$ is a finite \defd{set of states},
\ialphab~is a finite \defd{input alphabet},
$q_0\in Q$ is the \defd{initial state},
$F\subseteq Q$ is a \defd{set of final states},
and~%
$%
\delta:%
Q\times\ialphab%
\rightarrow%
\powerset{Q}%
$
is a \defd{nondeterministic transition function}.
At each step,
according to its current state~$p$
and the symbol~$\sigma$ scanned by the head,%
~\mA enters one nondeterministically-chosen state from~$\delta(p,\sigma)$
and moves the input head one position to the right.
The machine \defd{accepts} the input
if there exists a computation
starting from the initial state~$q_0$
with the head on the leftmost input symbol,
and ending in a final state $q\in F$
after having read the whole input
with the head to the right of the rightmost input symbol.
The language accepted by~\mA is denoted by~\defd{\langof\mA}.
An \nfa~\mA is said to be \defd{deterministic} (\defd\dfa)
whenever~$\card{\delta(q,\sigma)}\leq1$,
for any~$q\in Q$ and $\sigma\in\ialphab$.

The notions of configurations, successors, computations, and halting configurations,
previously introduced in the context of \dtm\s,
naturally transfer to \nfa\s.
\smallbreak

The \defd{size} of a machine
is given by the total number of symbols used to write down its description.
Therefore, the size of a one-tape Turing machine is bounded by a polynomial
in the number of states and of working symbols.
More precisely,
the device is fully represented by its transition function
which can be written in size~$\thetaof{\card Q\cdot\card\walphab\cdot\log(\card Q\cdot\card\walphab)}$.
In the case of \nfa\s (\resp, \dfa\s),
since no writings are allowed
and hence the working alphabet is not provided,
the size is linear in the number of instructions and states,
which is bounded by a function quadratic (\resp, subquadratic) in the number of states
and linear in the number of input symbols.
In this case, the description has size~$\thetaof{\card\ialphab\cdot\card Q^2}$
(\resp, $\thetaof{\card\ialphab\cdot\card Q\cdot\log\card Q}$).

\section{Hennie Machines: Undecidability and Non-Recursive Trade-Offs}
\label{sec:undec}

In this section we investigate some basic properties of \dtm\s.
First of all, we prove that it cannot be decided whether an end-marked \dtm works in linear time or not.
As a consequence, it cannot be decided if a \dtm is a Hennie machine.
Since linear-time \dtm\s accept only regular languages~\cite{Hen65}, it is natural to investigate
the size cost of their conversion into equivalent finite automata.
Even in the restricted case of deterministic Hennie machines we obtain a ``negative'' result, by proving a non-recursive trade-off between the size of Hennie machines and that of the equivalent finite automata.

Let us start by proving the following undecidability result.

\begin{theorem}
\label{th:undec}
	It is undecidable whether an end-marked \dtm works in linear time.
\end{theorem}
\begin{proof}
	We show that the problem of deciding if a \dtm halts on the empty word~$\emptyword$ reduces to this problem.
	Let $\mTM=\struct{Q,\ialphab,\walphab,\delta,q_0,F}$ be a \dtm.
	Without loss of generality,
	assume that \mTM has a tape infinite only to the right.
	Construct an end-marked Turing machine~\mHM
	with the input alphabet $\set a$ as follows.
	Given an input~$v\in a^*$,
	\mHM starts to simulate~\mTM over~\emptyword.
	If, during the simulation, \mHM reaches the right endmarker,
	then it stops the simulation and performs additional~$\Theta({\length v}^2)$ computation steps.%
	\footnote{%
		This can be achieved,
		for example,
		by a sequence of steps which overwrites every tape cell of the initial segment
		with a special marker~$\sharp\notin\walphab$,
		moving the head to the right endmarker after each rewriting.
		More precisely,
		from the cell containing the right endmarker,
		\mHM moves its head to the preceding cell
		and overwrites the contents with the special marker.
		After that,
		it moves its head to the right endmarker,
		then moves it backward to the rightmost cell not containing the special marker
		and writes~$\sharp$,
		thus repeating the procedure until all tape cells
		but those containing the endmarkers have been overwritten with the special marker.
	}
	Otherwise,~\mHM continues the simulation of~\mTM and halts
	if~\mTM halts.
	One can verify that the construction yields the following properties.
	\begin{itemize}
		\item If~\mTM halts on~$\emptyword$ in time~$t$ visiting~$s$ tape cells,
			then~\mHM performs~$\bigoof{t}$ computation steps on any input of length greater than~$s$,
			while it performs~$\bigoof{t^2}$ steps on shorter inputs.
			In both cases, the time is bounded by a constant in the input length.
		\item If~\mTM does not halt on~\emptyword,
			then for any input~$v$ either the simulation reaches the right endmarker
			and then~\mHM performs further~$\Theta(\length v^2)$ computation steps,
			or it does not halt because~\mTM enters an infinite loop,
			without reaching such a tape cell.
			In both cases~\mHM is not a linear-time~\dtm.
	\end{itemize}
	This allows to conclude
	that~\mHM is a linear-time~\dtm
	if and only if~\mTM halts on input~\emptyword,
	which is known to be undecidable.
\end{proof}


We now show that the size
trade-off from linear-time \dtm to finite automata is not recursive. More precisely, we obtain a
non-recursive trade-off between the sizes of Hennie machines and finite automata.

\begin{theorem}
\label{th:nonRec}
	There is no recursive function bounding the size blowup when transforming \dhm\s to finite automata.
\end{theorem}
\begin{proof}
	We recall that a busy beaver is an~$n$-state deterministic Turing machine
	with a two-symbol working alphabet $\set{\tblank,1}$
	that,
	starting its computation over a blank tape,
	halts after writing the maximum possible number~$\bb(n)$
	of $1$'s for its number of states.
	The function $\bb(n)$ of the space used by an~$n$-state busy beaver
	is known to be non-recursive,
	\ie, it grows asymptotically faster than any computable function~\cite{Rad62}.

	Here we consider a modification of the busy beaver
	that operates on a semi-infinite tape (instead of bi-infinite)
	and starts the computation on the leftmost cell,
	according to~\cite{Wal82}.
	This variant
	defines a different function~$\bb(n)$ which is also non-recursive.
	
	For each $n>0$,
	let $w_n$ be the string over $\set{a}$ of length $\bb(n)$ and let~$L_n=\set{w_n}$.
	This language is accepted by an end-marked \dtm \mHM[n]
	with $\bigoof{n}$ states and $\bigoof{1}$ working tape symbols,
	which simulates a given $n$-state busy beaver ($n$-$\BB$)
	and accepts an input $w\in a^*$ if and only if the space used by $n$-$\BB$ equals $\length w$.
	When~$n$-$\BB$ uses more than~$\length w$ space,
	at some point during the simulation the right endmarker is reached.
	At that point the simulation is aborted and the machine rejects.
	Furthermore, the simulation of~$n$-$\BB$ does not depend on the input.
	Hence, it is made in constant time. 
	This allows to conclude that,
	with respect to the input length,
	\mHM[n] works in linear time,
	so it is a Hennie machine.

	On the other hand,
	it is not difficult to see that the minimum \dfa accepting $L_n$ contains a path of~$S(n)+1$ states.
	This completes the proof.
\end{proof}

\section{Weight-Reducing Machines: Decidability, Expressiveness and Descriptional Complexity}
\label{sec:wrtm-versus-wrhm}
%

In \cref{sec:undec} we proved that it cannot be decided whether an end-marked \dtm works in linear time.
In this section we show that it is possible to decide whether \dtm\s are weight reducing or not.
Furthermore,
every linear-time \dtm~\mTM with the length of each computation bounded by $Kn+C$,
where $K,C$ are constants and~$n$ denotes the input length,
can be transformed into an equivalent weight-reducing machine
whose size is bounded by a recursive function of $K$ and the size of~\mTM. 

We also present a simulation of weight-reducing machines by finite automata, thus concluding that weight-reducing machines 
express exactly the class of regular languages.
From such a simulation, we will obtain the size trade-off between weight-reducing machines and finite automata which, hence,
is recursive.
This contrasts with the non-recursive trade-off from Hennie machines to finite automata, proved in \cref{sec:undec}.

\begin{proposition}
	\label{prop:is_wr?}
	It is decidable
	whether a \dtm is weight-reducing or not.
\end{proposition}
\begin{proof}
	Let $\mTM=\struct{Q,\ialphab,\walphab,\delta,q_0,F}$ be a \dtm.
	To decide if there is any order~$<$ on~\walphab
	proving that \mTM is weight-reducing,
	it suffices to check whether the directed graph $G=\struct{\walphab,E}$,
	with
	\[E=\set{(\tau,\sigma) \mid \exists p,q\in Q \,\, \exists d\in\moveset: \delta(p,\sigma)=(q,\tau,d)}\text,\]
	is acyclic
	(each topological ordering of~$G$ acts as the required order~$<$).
\end{proof}

We now study how linear-time \dtm\s can be made weight-reducing.
To this end,
we use the fact that each \dtm working in linear time makes a constant number of visits to each tape cell,
hence linear time implies a constant number of visits per tape cell.
This property is stated in the following lemma,
which derives from~\cite[Proof of Theorem 3]{Hen65}.

\begin{lemma}
	\label{lemma:linear->constant}
	Let $\mTM=\struct{Q,\ialphab,\walphab,\delta,q_0,F}$ be a \dtm.
	If there exist two constants~$K$ and~$C$
	such that every computation of \mTM
	has length bounded by $Kn+C$,
	where~$n$ denotes the input length,
	then \mTM never visits a tape cell more than $2K\cdot(\card Q)^K+K$ times.
\end{lemma}

The following lemma,
which will be used in this section to study trade-offs
between the computational models we are investigating and finite automata,
presents a transformation from linear-time Turing and Hennie machines
into equivalent weight-reducing ones.

\begin{lemma}\label{lemma:constantly-many-visits}
	Let $\mTM=\struct{Q,\Sigma,\Gamma,\delta,q_0,F}$ be a 
	\dtm such that, for any input, \mTM performs at most $k$ computation steps on each tape cell.
	Then there is a \wrdtm $\mA$ accepting $L(\mTM)$
	with the same set of states~$Q$ as~\mTM
	and working alphabet of size $\bigoof{k\cdot\card\Gamma}$.
	Furthermore,
	on each input \mA uses the same space as~\mHM.
	Hence,
	if~$\mTM$ is linear time or end-marked
	then so is~$\mA$.
\end{lemma}
\begin{proof}
	To obtain~$\mA$, we incorporate a counter into the working alphabet of~\mTM.
	For each scanned cell, the counter says
	what is the maximum number of visits $\mA$ can perform
	during the remaining computation steps over the cell.
	More formally, denoting by $\Sigma_{\tblank\,}$ the alphabet~$\Sigma\cup\{\tblank\}$,
	we define $\mA=\struct{Q,\Sigma,\Gamma',\delta',q_0,F}$
	with $\Gamma'=\Sigma_{\tblank\,}\cup\left((\Gamma\setminus\Sigma_{\tblank\,})\times\left\{0,\ldots,k-1\right\}\right)$
	and, for all $q,q'\in Q$, $a,a'\in\Gamma$, $d\in\{-1,+1\}$
	where $\delta(q,a)=(q',a',d)$, $\delta'$ fulfils
	\begin{alignat*}{3}
		&\delta'\left(q,a\right)  \,=\, \left(q',\left(a',k-1\right),d\right), && \quad\text{if }a\in\Sigma_{\tblank\,},\\
		&\delta'\left(q,\left(a, i\right)\right)  \,=\, \left(q',\left(a',i-1\right),d\right), && \quad\text{otherwise, for~}i=1,\ldots,k-1.
	\end{alignat*}
	Using an ordering~$<$ on $\Gamma'$ such that
	\begin{alignat*}{3}
		(a,i) & \,<\, b && \text{for all } a\in \left(\Gamma\setminus\Sigma_{\tblank\,}\right), b\in\Sigma_{\tblank\,}, i = 0,\ldots, k-1 \text{, and} \\
		(a,i) & \,<\, (b,j) \quad && \text{for all } a,b \in \left(\Gamma\setminus\Sigma_{\tblank\,}\right), i,j=0,\ldots, k-1, i < j\text,   
	\end{alignat*}
	it is easy to see that $\mA$ is a \wrdtm equivalent to $\mTM$.
	Furthermore, there is a natural bijection between computations of $\mTM$ and those of \mA, which preserves time (length of computations) and space (cells visited during the computation). Thus, if \mTM is linear time or end-marked, then so is \mA.
\end{proof}

By combining the above lemmas,
we obtain a procedure to convert linear-time Turing machines
into equivalent linear-time weight-reducing machines,
as soon as a linear time bound of the input device is explicitly given.
\begin{theorem}
	\label{thm:linear->weight-reducing}
	Let $\mTM=\struct{Q,\ialphab,\walphab,\delta,q_0,F}$ be a \dtm.
	If there exist two constants~$K$ and~$C$
	such that every computation of \mTM
	has length bounded by $Kn+C$,
	where~$n$ denotes the input length,
	then there is an equivalent linear-time \wrdtm
	with the same set of states~$Q$ as~\mTM
	and working alphabet of size $\bigoof{k\cdot\card\Gamma}$,
	where~$k=2K\cdot(\card Q)^K+K$.
\end{theorem}
\begin{proof}
	Direct consequence of \cref{lemma:linear->constant,lemma:constantly-many-visits}.
\end{proof}
\bigbreak

We now investigate the transformation of weight-reducing machines into equivalent finite automata and its cost.

\begin{theorem}\label{th:wrdtm-to-1dfa-upper-bound}
	For every \wrdtm $\mTM=\struct{Q,\Sigma,\Gamma,\delta,q_0,F}$ there exist an \nfa and a \dfa accepting $L(\mTM)$ with \mbox{$2^{\bigoof{\card\Gamma\cdot\log(\card Q)}}$} and  \mbox{$2^{2^{\bigoof{\card\Gamma\cdot\log(\card Q)}}}$} states, respectively.
\end{theorem}

\begin{proof}
	Assume $\mTM$ 
	always ends each accepting computation with the head scanning a tape cell to the right of the initial
	segment. This can be obtained, at the cost of introducing one extra symbol in the working alphabet,
	by modifying the transition function in such a way that when $\mTM$ enters a final
	state it starts to move its head to the right, ending when a blank cell is reached.
	Denote $n=\card Q$ and $m=\card\Gamma$.
	
	We describe an \nfa $\mA=\struct{Q',\Sigma,\delta',q_{\mathrm{I}},F'}$ which accepts $L(\mTM)$,
	working on the principle of guessing time-ordered sequences of states
	in which $\mTM$ scans each of the tape cells storing the input,
	together with the input symbol of the cell.
	This is a variant of the classical crossing sequence argument.
	In this case, for a tape cell~$C$,
	we consider the sequence of states in which the cell is scanned during a computation,
	while a crossing sequence is defined
	as the sequence of states of the machine when the border between two adjacent tape cells is crossed by the head.
	
	Suppose that the time-ordered sequence of states in which a cell~$C$ is scanned in a computation~$\rho$ is~$(q_1,\ldots,q_k)$.
	Due to the weight-reducing property,
	there are~$k$ or~$k+1$ the different contents of~$C$ in~$\rho$,
	depending on whether or not the computation stops in~$q_k$.
	Since the 
	working alphabet 
	consists of~$m$ symbols,
	we can conclude that~$k\leq m$.	
	
	The set of states $Q'$ thus consists of a special initial state $q_{\mathrm{I}}$,  
	a special final state $q_{\mathrm{F}}$, and all sequences of the form $(a,q_1,\ldots,q_k)$ where $a\in\Sigma \cup \{\tblank\}$, $1\le k \le m$, and $q_i\in Q$, for $i=1,\ldots,k$.
	
	
	Let $w\in\Sigma^+$ be a non-empty input.
	Let $\tau_l$, $\tau_{in}$ and $\tau_r$ denote the portion of \mTM's tape which initially stores the blank symbols preceding $w$, the input $w$, and the blank symbols to the right of~$w$, respectively.
	Let $\rho=(\mathcal{C}_0,\mathcal{C}_1,\ldots)$ be the computation of $\mTM$ over $w$.
	Let $q(j)$ denote the state of $\mTM$ in configuration~$\mathcal{C}_{j}$. Similarly, let $a(j)$ denote the symbol scanned by $\mTM$ in $\mathcal{C}_{j}$.
	For a tape cell $C$ in $\tau_{in}$, let $\mathcal{C}_{j_1},\ldots,\mathcal{C}_{j_k}$, where $j_1<\cdots<j_k$, be the sequence of all configurations in which \mTM scans $C$.
	Observe that $a(j_1)$ and $q(j_1),\ldots,q(j_k)$ determine $a(j_i)$ for all $i=2,\ldots,k$. For each configuration $\mathcal{C}_{j_i}$, it is also clear from which direction the head entered $C$ and in which direction it moves out of it ($\mathcal{C}_{j_1}$ is always entered from the left neighbouring cell, with the only exception of the initial configuration $\mathcal{C}_1$, which is indicated by the initial state $q_0$ that is never re-entered; for $i>1$, $\mathcal{C}_{j_i}$ is entered from the opposite direction than $\mathcal{C}_{j_{i-1}+1}$ was entered). For this reason, we can determine for two neighbouring cells $C_1$ and $C_2$ of $\tau_{in}$ whether two sequences of states assigned to them are consistent with $\mTM$ in the sense that the rightward head movements outgoing from $C_1$ have correspondent incoming leftward movements to $C_2$ and vice versa.
	Similarly, we can determine whether a sequence of states assigned to the first and to the last cell of $\tau_{in}$ is consistent with the computation of \mTM performed over the cell of $\tau_l$ and $\tau_r$, respectively.
	
	We are now going to formalize these ideas.
	
	Given $(a,q_1,\ldots,q_k), (b,p_1,\ldots,p_{\ell})\in Q'$, with~$a,b\in\Sigma$, let:
	\begin{itemize}
		\item $a_0=a$ and, for~$i=1,\ldots,k$, $\delta(q_i,a_{i-1})=(q'_i,a_i,d_i)$, $q'_i\in Q$, $a_i\in\Gamma$, $d_i\in\{-1,+1\}$;
		\item $b_0=b$ and, for~$i=1,\ldots,\ell$, $\delta(p_i,b_{i-1})=(p'_i,b_i,e_i)$, $p'_i\in Q$, $b_i\in\Gamma$, $e_i\in\{-1,+1\}$.
	\end{itemize}
	
	We say that~$(b,p_1,\ldots,p_{\ell})$ is \emph{consistent} with~$(a,q_1,\ldots,q_k)$ when
	there are indices~$i_1,i_2,\ldots,i_t$, $h_1,h_2,\ldots,h_t$, for some odd integer~$t\geq 1$,
	with~$1\leq i_1\leq i_2\leq\cdots\leq i_t=k$,
	$1=h_1\leq h_2\leq\cdots\leq h_t\leq\ell$,
	such that for~$j=1,\ldots,t$ it holds that:
	\begin{itemize}
		\item if~$j$ is odd then~$q'_{i_j}=p_{h_j}$, $d_{i_j}=+1$, and, when $j<t$, $i_j<i_{j+1}$,
		\item if~$j$ is even then~$p'_{h_j}=q_{i_j}$, $e_{h_j}=-1$, and~$h_j<h_{j+1}$,
	\end{itemize}
	while~$d_i=-1$ for~$i\notin\{i_1,\ldots,i_t\}$, and~$e_h=+1$ for~$h\notin\{h_1,\ldots,h_t\}$.

	Notice that, for any odd~$j$, in the transition from~$q_{i_j}$ to~$p_{h_j}$ the head crosses the border
	between the two adjacent cells by moving from left to right, while for any even~$j$ in the transition 
	from~$p_{h_i}$ to~$q_{j_i}$ the head crosses the same border by moving in the opposite direction.
	Furthermore, the conditions $t$ odd, $i_t=k$, $h_1=1$, derive from the fact that we are
	considering only sequences that could occur in accepting computations of~$\mTM$. In such
	computations, each cell of the initial segment
	is entered from the left  (with the exception of the leftmost one) and is finally left by moving the head to the right.
	(See \cref{fig:crossing} for an example).
\begin{figure}[tb]
	\centering
	\tikzset{%
	tape vlines/.append style={very thin,dotted},
	state/.append style={minimum size=1.25em,inner sep=0,font={\small},draw=none}
}

\begin{tikzpicture}
	\def\stateHdistance{1.5*\cellheight}
	\def\trackDistance{1.1*\cellwidth}
	\def\trackLangle{160}
	\def\trackRangle{20}
	\tape[0][window]{}[$a$,$b$]{2}[4][tape cell symb][1][1]
	
	\path
		(window cell 1.south)
			++(270:\stateHdistance)	node[state]	(q1) {$q_1$}
			++(270:\stateHdistance) node[state] (q2) {$q_2$}
			++(270:\stateHdistance) node[state] (q3) {$q_3$}
		(window cell 2.south)
			++(270:\stateHdistance)	node[state]	(p1) {$p_1$}
			++(270:\stateHdistance) node[state] (p2) {$p_2$}
			++(270:\stateHdistance) node[state] (p3) {$p_3$}
			++(270:\stateHdistance) node[state] (p4) {$p_4$}
			++(270:\stateHdistance) node[state] (p5) {$p_5$}
	;

	\draw[->,>=stealth]	(q2) -- (p1);
	\draw[->,>=stealth]	(p3) -- (q3);
	\draw[->,>=stealth]	(q3) -- (p4);

	\draw[->,>=stealth,thin,dashed]	(q1)++(\trackLangle:\trackDistance) -- (q1);
	\draw[->,>=stealth,thin,dashed]	(q1) -- ++(-\trackLangle:\trackDistance);

	\draw[->,>=stealth,thin,dashed]	(q2)++(\trackLangle:\trackDistance) -- (q2);

	\draw[->,>=stealth,thin,dashed]	(p1) -- ++(-\trackRangle:\trackDistance);
	\draw[->,>=stealth,thin,dashed]	(p2)++(\trackRangle:\trackDistance) -- (p2);
	\draw[->,>=stealth,thin,dashed]	(p2) -- ++(-\trackRangle:\trackDistance);

	\draw[->,>=stealth,thin,dashed]	(p3)++(\trackRangle:\trackDistance) -- (p3);
	\draw[->,>=stealth,thin,dashed]	(p4) -- ++(-\trackRangle:\trackDistance);
	\draw[->,>=stealth,thin,dashed]	(p5)++(\trackRangle:\trackDistance) -- (p5);
\end{tikzpicture}
	\caption{An example where~$(b,p_1,\ldots,p_5)$ is consistent with~$(a,q_1,q_2, q_3)$. Notice that $t=3$, $i_1=2$, $i_2=i_3=3$, $h_1=1$, $h_2=3$, $h_3=4$.}
\label{fig:crossing}

\end{figure}
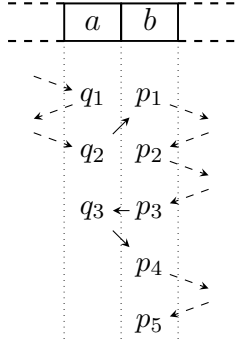
	\medskip
	
	In a similar way, we are going to identify the sequences from~$Q'$ which can occur on the
	rightmost cell of the initial segment in an \emph{accepting} computation
	(remember that we suppose that when~$\tau$ reaches a final state it starts to move
	its head to the right, ending when a blank cell is reached).
	Using the above notations, we say that \emph{the blank tape segment is consistent}
	with~$(a,q_1,\ldots,q_k)\in Q'$
	when there are indices~$1\leq i_1<i_2<\cdots<i_t=k$,
	and strings~$\gamma_0=\emptyword, \gamma_1, \ldots, \gamma_t\in\Gamma^*$, such that:
	\begin{itemize}
		\item $d_{i_j}=+1$, for~$j=1,\ldots,t$, while~$d_i=-1$ for~$i\notin\{i_1,\ldots,i_t\}$, namely,
		   the indices $i_j$ correspond to transitions moving the head to the right.
		\item For~$j=1,\ldots,t-1$, $\mTM$ in the state~$q_{i_j}$ with the head scanning a tape cell~$C$ 
		containing~$a_{{i_j}-1}$ and the string~$\gamma_{j-1}$ written on the non-blank cells to the right
		of~$C$, 
		moves to the right and makes a finite sequence of moves, which ends when the cell~$C$ is re-entered.
		At this point the state~$q_{i_j+1}$ and the string written on the non-blank cells to the right of~$C$
		is~$\gamma_j$.	
		\item From $q_{i_t}=q_k$ the machine~$\mTM$ moves its head to the right and, at some point,
		reaches a blank cell in a final state, without re-entering the cell~$C$ in between.
	\end{itemize}
	Let us denote by~$Q'_R$ the set of states~$(a,q_1,\ldots,q_k)\in Q'$ such that
	the blank tape segment is consistent with~$(a,q_1,\ldots,q_k)$.
	
\medskip
	
	Finally,
	we are now going to identify the sequences from~$Q'$ which are consistent with~$\tau_l$,
	namely sequences that could occur, in accepting computations, on the tape cell which initially 
	contains the leftmost input symbol.
	
	Given~$(a,q_1,\ldots,q_k)\in Q'$, with~$a\in\Sigma$ and, as before,
	$a_0=a$ and, for~$i=1,\ldots,k$, $\delta(q_i,a_{i-1})=(q'_i,a_i,d_i)$, $q'_i\in Q$, $a_i\in\Gamma$, $d_i\in\{-1,+1\}$,
	let~$1\leq i_1<i_2<\cdots<i_t<k$ be the indices corresponding to transitions
	moving the head to the left,
	\ie, $\{i_1,i_2,\ldots,i_t\}=\{i\mid d_i=-1\}$.
	We say that~$(a,q_1,\ldots,q_k)$ is \emph{consistent with the blank tape segment} when there are
	strings~$\gamma_0=\emptyword, \gamma_1, \ldots, \gamma_t\in\Gamma^*$, such that:
	\begin{itemize}
		\item $q_1$ is the initial state of~$\mTM$,
		\item for~$j=1,\ldots,t$, $\mTM$ in the state~$q_{i_j}$ with the head scanning a tape cell~$C$ 
		containing~$a_{{i_j}-1}$ and the string~$\gamma_{j-1}$ written on the non-blank cells to the left of~$C$, 
		moves to the left and makes a finite sequence of moves, up to re-enter~$C$.
		At that point the state is~$q_{i_j+1}$ and the string written on the non-blank cells to the left 
		of~$C$ is~$\gamma_j$.	
	\end{itemize}
	Let~$Q'_L$ denote the set of states from~$Q'$ which are consistent with the blank tape segment.
	
	\medskip

	At this point we developed all the tools to define the automaton~$\mA=\struct{Q',\Sigma,\delta',q_{\mathrm{I}},F'}$.
	\begin{itemize}
		\item The set of states is
			\[Q'=\{q_{\mathrm{I}},q_{\mathrm{F}}\}\cup\{(a,q_1,\ldots,q_k)\mid a\in\Sigma \cup \{\tblank\}, 1\le k \le m+1,  q_i\in Q,  i=1,\ldots,k\}\]
			as already mentioned.
		\item The initial state is $q_{\mathrm{I}}$.
		\item The transition function is defined,
			for $(a,q_1,\ldots,q_k)\in Q'$, $c\in\Sigma$, as
			\begin{eqnarray*}
			\lefteqn{\delta'((a,q_1,\ldots,q_k), c) =}\\		
			& &\left\{\begin{array}{ll}
				\{ (b,p_1,\ldots,p_{\ell}) \mid
				(b,p_1,\ldots,p_{\ell})\text{ is}&\\
				\mbox{~~consistent with }(a,q_1,\ldots,q_k) \}, 
				& \text{if } c=a\text{ and }(a,q_1,\ldots,q_k)\notin Q'_R\\
				\{ (b,p_1,\ldots,p_{\ell}) \mid
				(b,p_1,\ldots,p_{\ell})\text{ is}&\\
				\mbox{~~consistent with }(a,q_1,\ldots,q_k) \}\cup\{q_{\mathrm{F}}\}, 
				& \text{if } c=a\text{ and }(a,q_1,\ldots,q_k)\in Q'_R\\
				\emptyset,& \text{otherwise,}
			\end{array}\right.
			\end{eqnarray*}

			The transitions from the initial state~$q_{\mathrm{I}}$ are defined, for~$a\in\Sigma$,
			as
			\[\delta'(q_{\mathrm{I}},a)=\bigcup_{(a,q_1,\ldots,q_k)\in Q'_L}\delta'((a,q_1,\ldots,q_k), a)\,,\]
while there are no transitions from~$q_{\mathrm{F}}$, namely~$\delta'(q_{\mathrm{F}},a)=\emptyset$, for~$a\in\Sigma$.
		\item The set of final states is:
			\[F' =
			\left\{
				\begin{array}{ll}
					\{q_{\mathrm{F}}\}  & \text{if } \emptyword\notin{L(\mTM)} \\
					\{q_{\mathrm{I}},q_{\mathrm{F}}\},& \text{otherwise.}
				\end{array}\right.
				\]
	\end{itemize}
	
	By summarizing, \mA simulates~\mTM as follows:
	\begin{itemize}
		\item In the initial state~$q_{\mathrm{I}}$, reading an input symbol~$a$, $\mA$~implicitly
		guesses a sequence~$(a,q_1,\ldots,q_k)\in Q'_L$ and a sequence~$(b,p_1,\ldots,p_{\ell})$
		consistent with it, where~$b$ is supposed to be the symbol in the cell immediately to the right.
		The final state~$q_{\mathrm{F}}$ can be also guessed, if~$(a,q_1,\ldots,q_k)\in Q'_R$,
		\item When scanning an input cell containing a symbol~$a$, in a state $(a,q_1,\ldots,q_k)\in Q'$,
		$\mA$ guesses a sequence~$(b,p_1,\ldots,p_\ell)$ consistent with it.
		If the next input symbol is~$b$, then the simulation can continue in the same way, otherwise
		it stops because of an undefined transition.
		
		Furthermore, when $(a,q_1,\ldots,q_k)\in Q'_R$, $\mA$ can also guess to have reached the 
		last input symbol, so entering the final state~$q_{\mathrm{F}}$.
		If the end of the input is effectively reached then~$\mA$ accepts.
	\end{itemize}  
	 
	The number of states of~\mA is $2+ \left(\card\Sigma+1\right) \sum_{i=1}^mn^i=2^{\bigoof{m \log n}}$.
	If~\mA is in turn transformed to an equivalent \dfa,
	using the classical powerset construction\xspace
	,
	the resulting automaton has $2^{2^{\bigoof{m \log n}}}$ states.
\end{proof}

As a direct consequence of \cref{th:wrdtm-to-1dfa-upper-bound}, we get that \wrdtm\s
recognize exactly the class of regular languages.
\begin{corollary}
	A language is regular
	\ifof
	it is accepted by some~\wrdtm.
\end{corollary}

\Cref{th:wrdtm-to-1dfa-upper-bound} gives a double exponential upper bound for the size cost of the simulation of \wrdtm\s by \dfa\s. 
We now also prove a double exponential lower bound.

To this end, for each integer $n\geq0$,
we consider the language~$B_n$ over $\set{0,1,\$}$
consisting of strings $v_1\$v_2\$\cdots \$v_k$,
where~$k>2$,
$v_1,v_2,\ldots,v_k\in\set{0,1}^{\ast}$,
$\length{v_k}\leq n$,
$\length{v_i}\geq\length{v_k}$ for~$i=1,\ldots,k-1$,
and there exists~$j<k$ such that~$v_j=v_k$.
Informally, every string in $B_n$ is a sequence of binary blocks
which are separated by the symbol $\$$,
where the last block is of length at most~$n$
and it is a copy of one of the preceding blocks,
which all are at least as long as the last one.
For example,
$$v_1\$v_2\$v_3\$v_4\$v_5\$v_6=0011\$0101110\$011\$0011\$001\$011 \in B_4$$
since $\length{v_6}\le 4$, $\length{v_i}\geq\length{v_6}$ for~$i=1,\ldots,5$, and~$v_3=v_6$.

\begin{lemma}
\label{lemma:wrdtmBn}	
	For every integer $n \geq 0$,
	the language $B_n$ is accepted by a \wrdhm
	with~$\bigoof{1}$ states and~$\bigoof{n}$ working symbols.
\end{lemma}

\begin{proof}
	Let $\Sigma=\set{0,1,\$}$.
	We first describe an end-marked \dtm \mTM accepting the union of all~$B_i$'s, for~$i\geq0$,
	that has a constant number of states,
	then we show how \mTM can be modified in order to recognize~$B_n$, for a fixed integer~$n$,
	by bounding the number of visits to each cell,
	thus obtaining a \wrdhm with the desired properties.
	Let us define the working alphabet of \mTM as $\Gamma=\set{0,1,\$,x,f,\tblank}$.
	
	Let $w\in\Sigma^{\ast}$ be an input string of the form $w=v_1\$v_2\$\cdots \$v_k$,
	where~$v_1,\ldots,v_k\in\set{0,1}^{\ast}$,
	and $v_k=a_1\cdots a_{\ell}$, with $a_i\in \set{0,1}$ for~$i=1,\ldots,\ell$.
	The machine~\mTM performs~$\ell$ iterations.
	In each iteration it moves the head from the left endmarker
	to the right endmarker and back,
	thus visiting each input cell twice.
	It also rewrites some of the tape cells during this movement.
	The aim of the~$i$-th iteration is comparing the~$i$-th rightmost symbol of the last block
	with the~$i$-th rightmost symbol of any other block.
	This is implemented as follows.
	Within the first iteration,
	\mTM memorizes~$a_\ell$ in the states,
	rewrites it by the symbol~$x$,
	and moves the head leftwards.
	Whenever it encounters the symbol~$\$$ and enters the right end of a block~$v_j$,
	it checks if its last symbol equals~$a_{\ell}$.
	If so, \mTM overwrites the cell contents with~$x$, otherwise it writes~$f$.
	During the $i$-th iteration,
	\mTM memorizes~$a_{\ell+1-i}$ (which is in the rightmost input cell not containing the symbol $x$) in its finite control,
	overwrites the cell containing it by~$x$
	and checks whether the~$i$-th rightmost symbol of each $v_j$, with $j<k$, matches $a_{\ell+1-i}$
	(if so, it overwrites the symbol with~$x$, if not it writes~$f$).
	Notice that,
	at the beginning of the~$i$-th iteration, $i>1$,
	the~$i$-th rightmost symbol of a block is located immediately to the left
	of a nonempty factor consisting only of symbols~$x$ and~$f$.
	However, it could happen that for some~$j<k$ there is no $i$-th rightmost symbol in the factor~$v_j$,
	namely the block~$v_j$ is shorter than~$v_k$.
	In this case the machine halts and rejects.
	The input~$w$ is accepted by~\mTM
	\ifof, after some iteration,
	all symbols of~$v_k$ have been overwritten with~$x$
	and there is some~$v_j$ with all symbols also rewritten to~$x$
	(this ensures~$v_j=v_k$).
	A constant number of states is sufficient to implement the procedure so far described.


	It can be noticed that,
	for any fixed integer~$n$,
	a word belongs to~$B_n$ \ifof it is accepted by~\mTM within the first~$n$ iterations.
	Hence,
	as each iteration yields exactly two visits to each input cell,
	by bounding the number of visits to each cell by~$2n$,
	we can restrict~\mTM to accept words from~$B_n$ only.
	This can be obtained by using a construction similar as those used for proving \cref{lemma:constantly-many-visits}.
	We thus obtain a halting \wrdhm~\mHM accepting~$B_n$,
	which has~$\bigoof1$ states and~$\bigoof n$ working symbols.
\end{proof}

\begin{lemma}
\label{lemma:dfaBn}
	Each \dfa accepting $B_n$ has at least $2^{2^n}$ states.
\end{lemma}
\begin{proof}
	Let $\mathcal{S}$ be the family of all subsets of $\{0,1\}^n$.
	Given a subset $S=\{w_1,\ldots,w_k\}\in \mathcal{S}$,
	where $w_1<\cdots < w_k$ in the lexicographical order,
	consider the string $w(S)=w_1\$w_2\$\cdots \$w_k$.
	Let~$S_1$ and~$S_2$ be two different elements of~$\mathcal{S}$
	and let~$u\in\{0,1\}^n$ be a string which is in $S_1$ but not in $S_2$ (or vice versa).
	Then, $w(S_1)\$u\in B_n$ and $w(S_2)\$u\notin B_n$ (or vice versa),
	hence $\$u$ is a distinguishing extension, and, by the Myhill-Nerode Theorem,
	each \dfa accepting $B_n$ has at least $\card\mathcal{S}=2^{2^n}$ states.
\end{proof}

From \cref{th:wrdtm-to-1dfa-upper-bound,lemma:wrdtmBn,lemma:dfaBn}, we obtain:
\begin{corollary}
\label{cor:tradeOffs}
	The size trade-offs from \wrdtm\s and \wrdhm\s to \dfa\s are double exponential.
\end{corollary}

As shown in \cref{th:nonRec}, by dropping the weight-reducing assumption for machines, 
the size trade-offs in \cref{cor:tradeOffs} become not recursive.
However, provided an explicit linear bound on computation lengths,
we obtain the following result:
\begin{corollary}
	Let $\mTM=\struct{Q,\ialphab,\walphab,\delta,q_0,F}$ be a \dtm.
	If there exist two constants~$K$ and~$C$
	such that every computation of \mTM
	has length bounded by $Kn+C$,
	where~$n$ denotes the input length,
	then there exist an \nfa and a \dfa accepting $L(\mTM)$
	with \mbox{$2^{\bigoof{k\cdot\card\Gamma\cdot\log(\card Q)}}$}
	and \mbox{$2^{2^{\bigoof{k\cdot\card\Gamma\cdot\log(\card Q)}}}$} states, respectively,
	where~$k=2K\cdot(\card Q)^K+K$.
\end{corollary}
\begin{proof}
   Consequence of \cref{thm:linear->weight-reducing,th:wrdtm-to-1dfa-upper-bound}.
\end{proof}


\section{Weight-Reducing Machines: Space and Time Usage, Haltingness}
\label{sec:spaceTimeHalt}

Weight-reducing Turing machines generalize
weight-reducing end-marked Turing machines
(that are necessarily Hennie machines)
by allowing to use additional tape cells
that initially do not contain the input
and to which we refer as \defd{initially-blank cells}.
This extension allows in particular
infinite computations.
For instance, a~\wrdtm
can perform forward moves forever,
rewriting each blank cell with some symbol.
We now show that, however,
due to the weight-reducing property,
the amount of initially-blank cells
that is really useful,
\ie, that is visited in some halting computation,
is bounded by some constant
which can be computed from the size of the~\wrdtm
and does not depend on the input string.
This allows us to transform any \wrdtm
into an equivalent halting one of polynomial size,
which therefore operates in linear time.
Notice that \cref{th:wrdtm-to-1dfa-upper-bound}
already gave a simulation of \wrdtm\s
by a halting and linear-time computational model.

\begin{lemma}
\label{lemma:sameSequence}
	Each computation of a \wrdtm~\mTM which visits in the same sequence of states 
	two initially-blank cells of the tape, both located at the same side of the initial segment, is infinite.
\end{lemma}
\begin{proof}
	We give the proof in the case the two cells are located to the right of the initial segment.
	The proof for the other case can be obtained with a similar argument.
	For ease of exposition,
	we index the cell positions by integers,
	starting with the leftmost cell of the initial segment, whose index is~$1$,
	and we identify each cell with its position.
	
	Let us consider a configuration~$\mathcal C$ of~\mTM in which the head is scanning a tape cell~$c$,
	located in the portion of the tape to the right of the initial segment,
	containing a symbol~$a\in\Gamma$, the non-blank string written in the cells to the right of~$c$, starting from the cell~$c+1$,
	is~$\gamma\in(\Gamma\setminus\set{\tblank})^*$,\footnote{%
	Notice that when the cell~$c+1$ contains~$\tblank$, each cell~$c+h$ for~$h>0$ is not yet
	visited, hence~$\gamma=\emptyword$. For the same reason~$a=\tblank$ implies~$\gamma=\emptyword$.}
	and the state is~$q$.
	Thus~$\mathcal C=\config zq{a\gamma}$ for some~$z\in\Gamma^*$.
	Let us suppose that~$\delta(q,a)=(q',a',\xmove)$, with~$q'\in Q$, $a'\in\Gamma$, and~$\xmove\in\set{\lmove,\rmove}$.
	
	If~$\xmove=\rmove$, let us denote by~$path_R(q,a\gamma)$ the longest computation path which starts in the configuration~$\mathcal C$
	and, after~$\mathcal C$, visits only cells to the right of~$c$, possibly re-entering
	the cell~$c$ at the end. Notice that if the path does not re-enter the cell~$c$ then it could be
	infinite. Since~$\mTM$ is deterministic, we can observe the following facts:
	\begin{enumerate}[label={(\arabic*)},ref={(\arabic*)},nosep]
		\item\label{it1}The non-blank string which is written on the tape, 
		starting from cell~$c+1$, after the execution of~$path_R(q,a\gamma)$, if ending, only depends on~$q$, $a$, and~$\gamma$.
		\item\label{it2}For any fixed integer~$h>0$, the finite sequence of states which are reached when the head visits 
		the cell~$c+h$ during~$path_R(q,a\gamma)$, only depends on~$q$, $a$, $\gamma$, and~$h$.	
		This is also true in the case~$path_R(q,a\gamma)$ is infinite.
		Indeed, due to the fact that~\mTM is weight reducing, each cell can be visited only a finite number of times.
	\end{enumerate}
	These two facts will be now used in order to study 
	a computation visiting~$c$ in 
	configurations~$\mathcal C_0, \mathcal C_1,\ldots, \mathcal C_{k-1}$, $k\geq 1$.
	Let~$q_i, a_i, \gamma_i$ be
	the state, the symbol written in~$c$ and the non-blank contents of the cells to the right of~$c$
	in configuration~$\mathcal C_i$, respectively, $i=0,\ldots,k-1$.
	We are going to prove that, for each integer~$h>0$, the sequence~$P$ of states in which the cell~$c+h$
	is visited only depends on~$q_0,q_1,\ldots,q_{k-1}$ and~$h$. 
	
	Since the cell~$c+h$ cannot be visited before the cell~$c$, namely before configuration~$\mathcal C_0$,
	the sequence~$P$ can be decomposed as~$P=P_1P_2\cdots P_k$ where, for~$i=1,\ldots,k$, $P_i$ is the sequence of states which are reached when the head is visiting the cell~$c+h$ after the configuration~$\mathcal C_{i-1}$
	and, for~$i<k$, before the configuration~$\mathcal C_i$.
	
	Let us start by proving the following claims, for~$i=1,\ldots,k$:
	\begin{enumerate}[label={(C\arabic*)},ref={(C\arabic*)},nosep]
		\item\label{claim1}$P_i$ only depends on~$q_{i-1}$, $a_{i-1}$, $\gamma_{i-1}$, and~$h$,
		\item\label{claim2}if~$i<k$  then $a_i$ and~$\gamma_i$ only depend on~$q_{i-1}$, $a_{i-1}$ and~$\gamma_{i-1}$.
	\end{enumerate}
	To this end, first we observe that when~$i<k$,
	$\delta(q_{i-1},a_{i-1})=(q',a_i,d)$, for some state~$q'$, $d\in\{-1,+1\}$.
	There are two possibilities:
	\begin{itemize}
		\item $\xmove=\rmove$:
			in this case the cell~$c$ is re-entered in the state~$q_i$,
			\ie, $q_i$ is the last state in~$path_R(q_{i-1}, a_{i-1}\gamma_{i-1})$.
			Then~$\gamma_i$ is the string which is written in the cells to the right of~$c$ after executing~$path_R(q_{i-1}, a_{i-1}\gamma_{i-1})$ and, by~\ref{it1}, it depends only on~$q_{i-1}$, $a_{i-1}$, and~$\gamma_{i-1}$.
			Furthermore, by~\ref{it2}, the (possibly empty) sequence~$P_i$ of states in which the cell~$c+h$ 
			is visited along this path
			depends only on~$q_{i-1}$, $a_{i-1}$, $\gamma_{i-1}$, and~$h$.
		
		\item $\xmove=\lmove$:
			in this case, after a path which visits only cells to the left of~$c$,
			the cell~$c$ is re-entered in state~$q_i$.
			Since the contents of the cells to the right of~$c$ is not changed we have~$\gamma_i=\gamma_{i-1}$.
			Furthermore, $P_i$ is the empty sequence because the cell~$c+h$ was not reached in this path.
	\end{itemize}
	Even for~$i=k$, let~$\delta(q_{k-1},a)=(q',a',\xmove)$. We also have two possibilities:
	\begin{itemize}
		\item $\xmove=\rmove$:
			$path_R(q_{k-1}, a_{k-1}\gamma_{k-1})$ can be finite or infinite.
			By~\ref{it2}, the finite sequence~$P_k$ of states in which it visits~$c+h$ only depends on~$q_{k-1}$, $a_{k-1}$, $\gamma_{k-1}$, and~$h$.
		\item $\xmove=\lmove$:
			no more visits to the cell~$c+h$ are performed. Hence~$P_k$ is the empty sequence.
	\end{itemize}
	This completes the proof of~\ref{claim1} and~\ref{claim2}. Using these two statements and the fact that~$a_0=\tblank$ 
	and $\gamma_0=\emptyword$ are fixed, by proceeding in an inductive way,
	we obtain that~$P_i$ only depends on states~$q_0,q_1,\ldots,q_{i-1}$ and on~$h$, $i=1,\ldots,k$.
	
	This allows us to conclude that the sequence~$P=P_1P_2\cdots P_k$ of states reached at the cell~$c+h$
	depends only on the sequence of states~$q_0,q_1,\ldots,q_{k-1}$
	reached at the cell~$c$ and on~$h$, as we claimed.
	
	Suppose now that~$P$ coincides with the sequence~$q_0,q_1,\ldots,q_{k-1}$ of states visited at the cell~$c$.
	By iterating the previous argument,
	the sequence of states which are reached in any cell~$c+hj$ with~$j>0$ is~$q_0,q_1,\ldots,q_{k-1}$,
	thus implying the statement of the lemma.
\end{proof}

\begin{lemma}
	\label{lem:wr_space-bound}
	Let~\mTM be an~$n$\=/state \wrdtm
	which uses~$g$ working symbols.
	A computation of~\mTM is infinite
	if and only if it visits~$(n+1)^g$ consecutive initially-blank cells,
	\ie, tape cells to the left or to the right
	of the initial segment.
\end{lemma}
\begin{proof}
	Since~\mTM  is weight reducing, the number of visits to each tape cell is bounded by a
	constant which, in turn, is bounded by~$g$.
	Thus, each infinite computation should visit infinitely many tape cells,
	hence at least~$(n+1)^g$ consecutive initially-blank cells.
	
	To prove the converse, let us consider a halting computation~$\rho$ of~\mTM
	over an input word of length~$\ell$.
	By \cref{lemma:sameSequence}, $\rho$ cannot visit two tape cells, 
	laying at the same side of the initial segment, in the same sequence of states.
	Since there are less than~$(n+1)^g$
	nonempty distinct sequences of states of length at most~$g$,
	we conclude that
	the number of consecutive initially-blank cells
	visited during the computation
	is less than~$(n+1)^g$.
\end{proof}

As a consequence of the above result,
we obtain the following dichotomy of computations of \wrdtm\s.
\begin{proposition}
	\label{prop:wrdtm computation dichotomy}
	Each computation of a \wrdtm~\mTM
	either is infinite
	and visits an infinite amount of tape cells,
	or is finite,
	has length linearly bounded in the input length,
	and visits at most~$C$ initially-blank cells,
	for some constant~$C$ which depends only on~\mTM.
\end{proposition}
\begin{proof}
	Let~\mTM be an $n$-state \wrdtm having~$g$ working symbols,
	and let~$\rho$ be a computation of~\mTM over some input~$w$.
	If the amount of tape cells visited by~$\rho$ is~$k$, for some finite~$k$,
	then, as a cell cannot be visited more than~$g$ times
	by the weight-reducing property,
	$\rho$ is finite and has length bounded by~$gk$.
	Now, using \cref{lem:wr_space-bound}
	we have that~$\rho$ visits less than~$(n+1)^g$ initially-blank cells
	to the left (\resp, to the right) of the initial segment.
	Thus, $k<2(n+1)^g-1+\length w\in\bigoof{\length w}$.
	Conversely, if~$\rho$ visits infinitely many tape cells
	then it is necessarily infinite.
\end{proof}

\begin{proposition}
	\label{prop:wrtm->halting_wrtm}
	By a polynomial size increase,
	each \wrdtm can be transformed
	into an equivalent linear-time \wrdtm.
\end{proposition}
\begin{proof}
	From an~$n$-state \wrdtm~\mTM with a working alphabet of cardinality~$g$,
	we can build an equivalent halting \wrdtm~\mTM* which works as follows.
	After an initial phase
	during which~\mTM* marks~$(n+1)^g$ initially-blank cells
	to the left and to the right of the initial segment,
	it performs a direct simulation of~\mTM
	while controlling that no further cells
	than those initially marked and those of the initial segment
	are used.

	The initial phase is implemented
	using a counter in basis~$(n+1)$,
	stored on~$g$ consecutive tape cells,
	which is incremented up to~$(n+1)^g$
	and shifted along the tape.
	We shall describe a procedure marking the $(n+1)^g$ cells
	to the left of the initial segment.
	A similar procedure is repeated at the right of the initial segment.

	At each step, the counter contains the number of marked cells minus one.
	Hence, at the beginning, it is initialized to value~$g-1$ (in basis~$n+1$)
	by writing the corresponding digits onto the~$g$ cells immediately to the left of the initial segment,
	the least significant digit being on the rightmost of these cells.
	Then the counter is incremented and shifted leftward by updating each digit,
	from right to left, namely starting from the least significant one.
	Let~$d$ be the digit scanned by the head.
	The value~$d'=d+1\mod(n+1)$ is stored in the state control,
	along with a boolean variable~$c$,
	for carry propagation.
	Then the head is moved one cell to the left,
	$d'$ is written on the tape
	and its value updated according to the previous contents of the cell just overwritten
	and the value of~$c$\xspace%
	.
	After updating each of the~$g$ digits,
	the head is moved~$g$ positions to the right (on the least significant digit)
	and the counter is incremented again.
	This procedure stops when,
	moving rightward to reach the least significant digit,
	all the~$g$ scanned cells contain the symbol~$n$.
	It is possible to notice that
	the number of states used to implement this procedure is~$\bigoof{g+n}$.
	Furthermore,
	the number of visits to each initially-blank cell
	during this procedure is bounded by~$2g$
	because every cell is visited at most twice
	when containing the~$i$-th digit of the current counter value,
	for~$i=1,\ldots,g$.
	Thus, using \cref{lemma:constantly-many-visits},
	we can implement the procedure using
	a halting weight-reducing machine
	that uses~$\bigoof{gn}$ symbols.

	Once the space is marked,
	$\mTM'$ simulates~$\mTM$,
	stopping and rejecting if the simulation reaches a blank cell.
	Since~$\mTM$ is weight-reducing
	$\mTM'$ is weight-reducing as well.
	
	From \cref{lem:wr_space-bound}, we can easily conclude that
	each halting computation of~$\mTM$ is simulated by an equivalent
	halting computation of~$\mTM'$, while each infinite computation
	of~$\mTM$ is replaced in~$\mTM'$ by a computation which reaches
	a blank cell and then stops and rejects.
	Moreover,~$\mTM'$ uses~$\bigoof{g+n}$ states
	and~$\bigoof{gn}$ working symbols.
\end{proof}

Using \cref{lem:wr_space-bound,prop:wrtm->halting_wrtm},
we prove the following property.
\begin{theorem}
	\label{th:halting}
	It is decidable whether a \wrdtm halts on each input string.
\end{theorem}
\begin{proof}
	From any given \wrdtm~$\mTM$, we construct a halting \wrdtm~$\mTM'$ which, besides all the strings accepted by~$\mTM$, accepts all the strings
	on which~$\mTM$ does not halt. To this end, we can slightly modify the construction used to prove \cref{prop:wrtm->halting_wrtm},
	in such a way that when the head reaches a blank cell outside the initial segment and the initially marked space, the machine stops and accepts.
	Hence, the given \wrdtm~$\mTM$ halts on each input string if and only if the finite automata which are obtained from~$\mTM$ and~$\mTM'$ 
	according to \cref{th:wrdtm-to-1dfa-upper-bound} are equivalent.
\end{proof}

As a consequence:
\begin{corollary}
	It is decidable whether a \wrdtm works in linear time.
\end{corollary}
\end{document}

\section{Conclusion}
\label{sec:conclusion}

In this work,
we investigated deterministic one-tape Turing machines working in linear time.
Although these devices are known to be equivalent to finite automata~\cite{Hen65},
one cannot decide whether a given Turing machine is linear-time
even in the case of end-marked devices,
as we showed in \cref{th:undec}.
Furthermore, there is no recursive function bounding
the size blowup of the conversion of \dhm\s into \dfa\s (\cref{th:nonRec}).
To avoid these negative results,
we introduced the weight-reducing restriction,
that forces one-tape Turing machines
to work in linear time as long as they halt.
Indeed, we proved that each computation of a \wrdtm
either is infinite
or halts within a linear number of steps in the input length (\cref{prop:wrdtm computation dichotomy}).
The weight-reducing restriction is syntactic and can be checked (\cref{prop:is_wr?}).
Furthermore,
we proved in \cref{th:wrdtm-to-1dfa-upper-bound}
that each \wrdtm can be converted into an \nfa (\resp, \dfa)
whose size is exponential (\resp, doubly-exponential)
with respect to the size of the converted device.
These costs are tight~(\cref{cor:tradeOffs}).

Weight-reducing Turing machines
are not restrictions of linear-time machines.
Indeed, they allow infinite computations.
However, the haltingness of \wrdtm\s on any input can be decided (\cref{th:halting}).
Furthermore,
with a polynomial increase in size,
each weight reducing machine can be made halting and linear-time
(\cref{prop:wrtm->halting_wrtm}).
Still, \emph{halting} \wrdtm\s are not particular \dhm\s
as they allow the use of extra space
besides the initial segment
contrary to \dhm\s which are end-marked.
We do not know at the time of writing
whether this extra space usage
is useful for concisely representing regular languages.
In other words,
we leave open
the question of the size cost of turning \wrdtm\s
into equivalent \wrdhm\s.

In a related paper~\cite{GPPP21b},
we continue the investigation of computational models considered here
by focusing on the Sakoda and Sipser question
about the size cost of the determinization of two-way finite automata~\cite{Sak78}.
We indeed propose a new approach to this famous open problem,
which consists in converting two-way nondeterministic automata
into equivalent deterministic extensions of two-way finite automata,
paying a polynomial increase in size only.
The considered extensions are variants of linear-time deterministic Turing machines,
including \wrdhm\s and \dhm\s.

\bibliography{wrhm}
\end{document}